\documentclass[a4paper,11pt,fleqn,draft]{article}%

\usepackage{amsmath,amssymb,amsthm} 
\usepackage[mathscr]{eucal}
\usepackage{url}

\allowdisplaybreaks
\newcommand{\todo}[1][\null]{\ensuremath{\clubsuit}}

\newcommand{\noprint}[1]{}
\newcommand{\checked}[1][\null]{\ensuremath{\boldsymbol{\surd}}}

\newcommand{\p}{\partial}

\newcommand{\ri}{\mathfrak r}

\newtheorem{theorem}{Theorem}

{\theoremstyle{definition}
\newtheorem{remark}[theorem]{Remark}
}

\begin{document}

\par\noindent {\LARGE\bf
On the prolongation of a local hydrodynamic-type \\Hamiltonian operator to a nonlocal one
\par}

\vspace{4mm}\par\noindent {\large Stanislav Opanasenko$^{\dag\ddag}$ and Roman O.\ Popovych$^{\ddag\S}$
}

\vspace{4mm}\par\noindent {\it\small
$^{\dag}$~Sezione INFN di Lecce, via per Arnesano, 73100 Lecce, Italy\par
}

\vspace{1mm}\par\noindent{\it\small
$^\S$\,Mathematical Institute, Silesian University in Opava, Na Rybn\'\i{}\v{c}ku 1, 746 01 Opava, Czech Republic\par
}

\vspace{1mm}\par\noindent{\it\small
$^\ddag$\,Institute of Mathematics of NAS of Ukraine, 3 Tereshchenkivska Str., 01024 Kyiv, Ukraine\par
}

\vspace{1mm}\par\noindent
\textup{E-mail:} stanislav.opanasenko@le.infn.it, rop@imath.kiev.ua\!
\par

\vspace{6mm}\par\noindent\hspace*{9mm}\parbox{142mm}{\small
Nonlocal Hamiltonian operators of Ferapontov type are well-known objects
that naturally arise local from Hamiltonian operators of Dubrovin--Novikov type with the help of three constructions,
Dirac reduction, recursion scheme and reciprocal transformation.
We provide an additional construction, namely the prolongation of a local hydrodynamic-type Hamiltonian operator of a subsystem
to its nonlocal counterpart for the entire system.
We exemplify this construction by a system governing an isothermal no-slip drift flux.
}\par

\noprint{


76-XX   Fluid mechanics {For general continuum mechanics, see 74Axx, or other parts of 74-XX}
 76Mxx  Basic methods in fluid mechanics [See also 65-XX]
  76M60 Symmetry analysis, Lie group and algebra methods

37-XX	Dynamical systems and ergodic theory
 37Kxx  Infinite-dimensional Hamiltonian systems [See also 35Axx, 35Qxx]
  37K05 Hamiltonian structures, symmetries, variational principles, conservation laws

35-XX   Partial differential equations
 35Bxx  Qualitative properties of solutions
  35B06 Symmetries, invariants, etc.

Potential reviewers: Sheftel, Ferapontov, Anco, Shevyakov, Grunland, Novikov, Olver, Oberlack

}

\section{Introduction}\label{sec:IDFMIntroduction}

Since the seminal paper~\cite{DubrovinNovikov1984} was published, hydrodynamic-type Hamiltonian operators
$A=(A^{ij})$ in dimension (1+1) with components
\begin{gather}\label{IDFM:eq:LocalHamOper}
A^{ij}=g^{ij}(u)\mathrm D_x+b^i_{jk}(u)u^k_x
\end{gather}
where $u(t,x)=\big(u^1(t,x),\dots,u^n(t,x)\big)$, $(t,x)\in\mathbb R^2$ and Einstein summation convention is assumed,
have been thoroughly studied in many directions.
Recall that the nature of such operators (as well as the corresponding Poisson brackets) is differential-geometric.
More specifically, under the condition $\det(g^{ij})\ne0$ and the representation \smash{$b^i_{jk}=-g^{is}\Gamma^j_{sk}$},
upon (nondegenerate) point transformations of dependent variables,
the coefficients $g^{ij}$ and \smash{$\Gamma^j_{sk}$}
transform as coefficients of a (pseudo-)Riemannian metric and as Christoffel symbols of an affine connection~$\nabla$, respectively.
Moreover, for the operator~$A$ of the form~\eqref{IDFM:eq:LocalHamOper} to be Hamiltonian,
the connection~$\nabla$ should be symmetric and compatible with the metric, while the metric itself should be flat.

The subject of the present paper is the nonlocal generalisation of the hydrodynamic-type Hamiltonian operators
in the following way:
\begin{gather}\label{IDFM:eq:NonlocalHamOper}
A^{ij}=g^{ij}(u)\mathrm D_x-g^{is}\Gamma^j_{sk}(u)u^k_x+\mathfrak T^{ij}_w
\quad\mbox{with}\quad
\mathfrak T^{ij}_w:=\sum\limits_{\alpha=1}^N
\epsilon_\alpha w^i_{\alpha k}(u)u^k_x\mathrm D^{-1}_x w^j_{\alpha l}(u)u^l_x,
\end{gather}
where $\epsilon_\alpha\in\{-1,1\}$,
$\alpha$ runs from~1 to~$N$, and $w_\alpha:=(w^j_{\alpha l})$ transform as $(1,1)$-tensors (affinors) upon point transformations of~$u$.
The Hamiltonian property of operators of the form~\eqref{IDFM:eq:NonlocalHamOper}
and their geometric interpretation were related in~\cite{Ferapontov1991a}.

\begin{theorem}[\cite{Ferapontov1991a}]\label{thm:NonlocalHamiltonianOpsWithAffinors}
An operator given by~\eqref{IDFM:eq:NonlocalHamOper} is Hamiltonian if and only if
the tensor~$g=(g^{ij})$ defines a (pseudo-)Riemannian metric,
the connection~$\nabla$ is symmetric and compatible with the metric, $\nabla_k g^{ij}=0$,
the tensor~$g$ and the affinors $w_\alpha:=(w^j_{\alpha l})$ satisfy the system
\[
g_{ik}w^k_{\alpha j}=g_{jk}w^k_{\alpha i},\quad \nabla_k w^i_{\alpha j}=\nabla_j w^i_{\alpha k},\quad
R^{ij}_{\ \ kl}=\sum\limits_{\alpha=1}^N \epsilon_\alpha(w^i_{\alpha l}w^j_{\alpha k}-w^i_{\alpha k}w^j_{\alpha l})
\]
and the affinors~$w_\alpha$ commute, $[w_\alpha,w_\beta]=0$, $\alpha,\beta=1,\dots,N$.
Here $R^{ij}_{\ \ kl}:=g^{is}R^{j}_{\ skl}$, and $R^{j}_{\ skl}$ is the curvature tensor of the metric.
\end{theorem}

The equations in Theorem~\ref{thm:NonlocalHamiltonianOpsWithAffinors}
are nothing else but the Gauss--Peterson--Mainardi--Codazzi equations for submanifolds~$M^n$
with a flat normal connection in a (pseudo-)Euclidean space of dimension~$n+N$.
The metric~$g$ plays the role of the first quadratic form of~$M^n$,
and $w_\alpha$ are the Weingarten operators corresponding to the field of pairwise orthogonal unit normals.
More specifically, the second set of equations in~Theorem~\ref{thm:NonlocalHamiltonianOpsWithAffinors} are Peterson--Mainardi--Codazzi equations,
and the third set is the Gauss equations of the hypersurface.

Nonlocal Hamiltonian operators of the form~\eqref{IDFM:eq:NonlocalHamOper} are known to appear
as a result of such manipulations with their local counterparts~\eqref{IDFM:eq:LocalHamOper}
as recursion schemes, actions of reciprocal transformations or the Dirac reduction,
see a thorough review in~\cite{Ferapontov1995}.
In physical applications, they were shown to appear
in the context of higher-dimensional Witten--Dijkgraaf--Verlinde--Verlinde (WDVV) equations~\cite{VasicekVitolo2021}.

A recursion scheme applies whenever a hydrodynamic-type system of PDEs
\begin{gather}\label{IDFM:eq:HydroTypeSys}
u^i_t=v^i_j(u)u^i_x
\end{gather}
admits two compatible Hamiltonian operators~$A_0$ and~$A_1$ of the form~\eqref{IDFM:eq:LocalHamOper},
where the compatibility means that any linear combination of~$A_0$ and~$A_1$
is also a Hamiltonian operator of~\eqref{IDFM:eq:HydroTypeSys}.
The recursion operator $R=A_1A_0^{-1}$ defines
an infinite sequence of compatible nonlocal Hamiltonian operators $A_k=RA_{k-1}$, $k\geqslant 1$,
of the form~\eqref{IDFM:eq:NonlocalHamOper}, see examples in~\cite{AntonowiczFordy1987,FerapontovPavlov1991}.

Let a hydrodynamic-type system~\eqref{IDFM:eq:HydroTypeSys} admit a Hamiltonian operator~$A$ of the form~\eqref{IDFM:eq:LocalHamOper}
and two conserved currents $(\rho_1,\sigma_1)$ and $(\rho_2,\sigma_2)$, that is,
$\mathrm D_t\rho_i+\mathrm D_x\sigma_i=0$ in view of~\eqref{IDFM:eq:HydroTypeSys}, $i=1,2$.
Then the pair of the system~\eqref{IDFM:eq:HydroTypeSys} and the Hamiltonian operator~$A$ is mapped by the reciprocal transformation
\[
\mathrm d\tilde t=\sigma_1\mathrm dt+\rho_1\mathrm dx,\quad \mathrm d\tilde x=\sigma_2\mathrm dt+\rho_2\mathrm dx,
\]
to a pair of a hydrodynamic-type system and its Hamiltonian operator,
which may be of the nonlocal form~\eqref{IDFM:eq:NonlocalHamOper},
see details in~\cite{Ferapontov1989,Ferapontov1991c} and Remark~\ref{rem:ReciprocalTrans} below.

The Dirac reduction allows one to restrict a local hydrodynamic-type Hamiltonian operator on the phase space~$(u^1,\dots,u^{n+N})$
to a nonlocal Hamiltonian operator of the form~\eqref{IDFM:eq:NonlocalHamOper} on its certain subspace~$(u^1,\dots,u^n)$, see details in~\cite{Ferapontov1992}.

In the present paper, we propose another method of obtaining a nonlocal Hamiltonian operator~\eqref{IDFM:eq:LocalHamOper}
that is in a sense opposite to the Dirac reduction.
Starting with a local Hamiltonian operator of a hydrodynamic-type system~$\mathcal S$,
we prolong it to a nonlocal operator of a supersystem of~$\mathcal S$.

\section{Example}\label{sec:IDFMNonlocalHamiltonianStructure}

We consider the diagonalised form
\begin{gather*}
\mathcal S\colon\quad
\ri^1_t+(\ri^1+\ri^2+1)\ri^1_x=0,
\quad
\ri^2_t+(\ri^1+\ri^2-1)\ri^2_x=0,
\quad
\ri^3_t+(\ri^1+\ri^2)\ri^3_x=0,
\end{gather*}
of the  hydrodynamic-type system
\begin{gather*}
\tilde{\mathcal S}\colon\quad
\rho^1_t+u\rho^1_x+u_x\rho^1=0,\quad
\rho^2_t+u\rho^2_x+u_x\rho^2=0,\quad
(\rho^1+\rho^2)(u_t+uu_x)+\rho^1_x+\rho^2_x=0,
\end{gather*}
which appears in the theory of two-phase flow phenomenon
as an isothermal no-slip drift flux model~\cite{EvjeFlatten2007}.
Here, the Riemann invariants~$\ri^1$, $\ri^2$ and~$\ri^3$ are defined by
\[
\ri^1=\frac{u+\ln(\rho^1+\rho^2)}2,\quad \ri^2=\frac{u-\ln(\rho^1+\rho^2)}2,\quad \ri^3=\frac{\rho^2}{\rho^1},
\]
and $\ri:=(\ri^1,\ri^2,\ri^3)$.
Any constraint meaning that $\rho^1$ and $\rho^2$ are proportional, e.g., $\rho^2=0$,
reduces the system~$\tilde{\mathcal S}$ to the system~$\tilde{\mathcal S_0}$
describing one-dimensional isentropic gas flows with constant sound speed,
cf. the system (3)--(4) with $\nu=0$ in~\cite[Section~2.2.7]{RozhdestvenskiiJanenko1983}.
The system~$\mathcal S$ is partially coupled.
Its essential subsystem~$\mathcal S_0$ consisting of its first two equations
coincides with the diagonalized form of the system~$\tilde{\mathcal S_0}$
\mbox{\cite[Section~2.2.7, Eq.~(16)]{RozhdestvenskiiJanenko1983}.}
The system~$\mathcal S_0$ is known to possess three compatible Hamiltonian structures of hydrodynamic type~\cite{Nutku1987},
\begin{gather*}
\mathfrak H^1=\mathrm e^{\ri^2-\ri^1}
\left(\begin{pmatrix}
-1 & 0 \\ 0 & 1
\end{pmatrix}\mathrm D_x-\frac12
\begin{pmatrix}
\ri^2_x-\ri^1_x & \ri^1_x-\ri^2_x \\
\ri^2_x-\ri^1_x & \ri^1_x-\ri^2_x
\end{pmatrix}\right),\\
\mathfrak H^2=\mathrm e^{\ri^2-\ri^1}
\left(\begin{pmatrix}
1 & 0 \\ 0 & 1
\end{pmatrix}\mathrm D_x+\frac12
\begin{pmatrix}
\ri^2_x-\ri^1_x & -\ri^1_x-\ri^2_x \\
\ri^1_x+\ri^2_x & \ri^2_x-\ri^1_x
\end{pmatrix}\right),\\
\mathfrak H^3=\mathrm e^{\ri^2-\ri^1}
\left(\begin{pmatrix}
\ri^1 & 0 \\ 0 & \ri^2
\end{pmatrix}\mathrm D_x+\frac12
\begin{pmatrix}
(1-\ri^1)\ri^1_x+\ri^1\ri^2_x & -\ri^2\ri^1_x-\ri^1\ri^2_x \\
\ri^2\ri^1_x+\ri^1\ri^2_x     & -\ri^2\ri^1_x+(1+\ri^2)\ri^2_x
\end{pmatrix}\right).
\end{gather*}
\noprint{
with the associated families of Hamiltonians
\begin{gather*}
\mathcal H^1_{c_1,c_2}=\left(\frac14(\ri^1+\ri^2)^2+\frac12\big(\ri^1-\ri^2\big)+c_1\right){\rm e}^{\ri^1-\ri^2}+c_2(\ri^1+\ri^2),\\
\mathcal H^2_{c_1,c_2}=\mathrm e^{(\ri^1-\ri^2)/2}\left(c_1\sin\frac{\ri^1+\ri^2}2+c_2\cos\frac{\ri^1+\ri^2}2\right)-(\ri^1+\ri^2)\mathrm e^{\ri^1-\ri^2},\\
\mathcal H^3_{c_1,c_2}=c_1\operatorname{erf}\left(\frac{\sqrt{\ri^2}+\sqrt{-\ri^1}}{\sqrt{2}}\right)+c_2\operatorname{erf}\left(\frac{\sqrt{\ri^2}-\sqrt{-\ri^1}}{\sqrt{2}}\right).
-2\mathrm e^{\ri^1-\ri^2}.
\end{gather*}

The multipliers of parameterizing constants~$c_1$ and~$c_2$ are densities of the Casimir functionals of the corresponding
Hamiltonian operators, and therefore, by and large, there is only one Hamiltonian for every family of Hamiltonian operators.
}

An infinite family of compatible Hamiltonian structures of the system~$\mathcal S$
that are parameterized by an arbitrary function of a single argument
was constructed in~\cite{OpanasenkoBihloPopovychSergyeyev2020b}.
In fact, the elements of this family exhaust all the local hydrodynamic-type Hamiltonian structures of~$\mathcal S$
and are local prolongations, to~$\mathcal S$, of the Hamiltonian structure for~$\tilde{\mathcal S_0}$
that is related to the operator~$\mathfrak H^1$.

\begin{theorem}\label{theorem:IDFMHamiltonianStructures}
The system~$\mathcal S$ admits an infinite family of compatible local Hamiltonian operators~$\mathfrak H^1_\Theta$
parameterized by ah arbitrary smooth function~$\Theta$ of~$\ri^3$,
\begin{gather*}
\mathfrak H^1_\Theta={\rm e}^{\ri^2-\ri^1}\left(\mathop{\rm diag}\big(-1,1,{\rm e}^{\ri^2-\ri^1}\Theta(\ri^3)\big)\mathrm D_x-\frac12
\begin{pmatrix}
\ri^2_x-\ri^1_x  &  \ri^1_x-\ri^2_x  &  -2\ri^3_x\\[1ex]
\ri^2_x-\ri^1_x  &  \ri^1_x-\ri^2_x  &  -2\ri^3_x\\[1ex]
2\ri^3_x         &  2\ri^3_x         &  f^{33}
\end{pmatrix}\right),
\end{gather*}
where $f^{33}:=-{\rm e}^{\ri^2-\ri^1}\left(2(\ri^2_x-\ri^1_x)\Theta+\ri^3_x\Theta_{\ri^3}\right)$.
\end{theorem}

Since there are no known local prolongations of the Hamiltonian operators~$\mathfrak H^2$ and~$\mathfrak H^3$
of the system~$\tilde{\mathcal S_0}$ to the system~$\mathcal S$,
we consider their nonlocal hydrodynamic-type prolongations $\mathcal N=(\mathcal N^{ij})$
with components of the form~\eqref{IDFM:eq:NonlocalHamOper}.

\begin{theorem}
The system~$\mathcal S$ admits two families of nonlocal first-order Hamiltonian operators of hyd\-rodynamic type,
\begin{gather*}
\hat {\mathfrak H}^2_{\Theta,\alpha}
=\mathrm e^{\ri^2-\ri^1}\left(\mathrm{diag}\big(1,1,\hat{\Theta}\big)\mathrm D_x
+\frac12
\begin{pmatrix}
\ri^2_x-\ri^1_x &-\ri^1_x-\ri^2_x&-2\ri^3_x\\
\ri^1_x+\ri^2_x &\ri^2_x-\ri^1_x &2\ri^3_x\\
2\ri^3_x&-2\ri^3_x& f^{33}
\end{pmatrix}\right)
+\mathfrak T_{\hat w},
\\[1ex]
\hat {\mathfrak H}^3_{\Theta,\alpha}
=\mathrm e^{\ri^2-\ri^1}\mathop{\rm diag}\big(\ri^1,\ri^2,\hat{\Theta}\big)\mathrm D_x
+\frac{\mathrm e^{\ri^2-\ri^1}}2\!
\begin{pmatrix}
\ri^1_x+\ri^1(\ri^2_x-\ri^1_x)\! & -\ri^2\ri^1_x-\ri^1\ri^2_x     &  -2\ri^1\ri^3_x\\
\ri^2\ri^1_x+\ri^1\ri^2_x     & \!\!\ri^2_x+\ri^2(\ri^2_x-\ri^1_x)\!\! &  2\ri^2\ri^3_x\\
2\ri^1\ri^3_x & -2\ri^2\ri^3_x& f^{33}
\end{pmatrix}
{+}\mathfrak T_{\check w},
\end{gather*}
where $\alpha=1,2,3$,
the components of the operator $\mathfrak T_w=(\mathfrak T_w^{ij})_{i,j=1,2,3}$ are of the form~\eqref{IDFM:eq:NonlocalHamOper}, 
\begin{gather*}
\hat w_\alpha=\mathop{\rm diag}\big(c_{\alpha} ,c_{\alpha},c_{\alpha}+\Phi^\alpha\mathrm e^{\ri^2-\ri^1}\big),\\
\check w_\alpha=c_\alpha\mathop{\rm diag}\big(\ri^1+\ri^2+1,\ri^1+\ri^2-1,\ri^1+\ri^2\big)+
\frac12\Phi^\alpha\mathrm e^{\ri^2-\ri^1}\mathop{\rm diag}\big(0,0,1\big),
\end{gather*}
$\epsilon_1=\epsilon_2=-\epsilon_3=1$, $\hat\Theta=\mathrm e^{\ri^2-\ri^1}\Theta$,
$f^{33}:=2(\ri^2_x-\ri^1_x)\hat\Theta+\hat\Theta_{\ri^3}\ri^3_x$,
$\Phi^\alpha=b_{1\alpha}\Lambda^1+b_{2\alpha}\Lambda^2+b_{3\alpha}$
with arbitrary smooth functions~$\Theta$, $\Lambda^1$ and~$\Lambda^2$ of~$\ri^3$ such that
$\Lambda^1$, $\Lambda^2$ and~the constant function~$1$ of~$\ri^3$ are linearly independent,
$c_\alpha$, $b_{1\alpha}$, $b_{2\alpha}$ and~$b_{3\alpha}$ are arbitrary constants satisfying the system
\[
\sum_\alpha \epsilon_\alpha c_\alpha^2=0,\quad
\sum_\alpha \epsilon_\alpha c_\alpha b_{i\alpha}=0,\ i=1,2,\quad
\sum_\alpha \epsilon_\alpha c_\alpha b_{3\alpha}=-1.
\]
\end{theorem}

\begin{proof}
First of all, the restriction of the local part of the nonlocal operator~$\mathcal N$ of the form~\eqref{IDFM:eq:NonlocalHamOper}
to $(\ri^1,\ri^2)$ should coincide with the operator to be prolonged,~$\mathfrak H^2$ or~$\mathfrak H^3$.
To find further restrictions on~$\mathcal N$, we employ the fact that any Hamiltonian operator is Noether,
that is, it maps cosymmetries of the system under study into its symmetries.

Denote the universal linearization operator~\cite{Bocharov1999}
(or, in other terminology, the Fr\'echet derivative~\cite{Olver1993})
of the left hand side~$S$ of the system~$\mathcal S$ and its adjoint by~$\ell_S$  and \smash{$\ell^\dag_S$},
respectively.
Then the tangent and the cotangent coverings of the system~$\mathcal S$
are the system $T\mathcal S$: $S=0$, $\ell_S(\eta)=0$ and $T^*\mathcal S$: $S=0$, \smash{$\ell^\dag_S(\lambda)=0$},
where $\eta=(\eta^1,\eta^2,\eta^3)^{\mathsf T}$ and $\lambda=(\lambda_1,\lambda_2,\lambda_3)$
are additional dependent variables.
The operator~$\mathcal N$ is a Noether one if $(\ri,\eta)$ with $\eta^i=\mathcal N^{ij}\lambda_j$
is a solution of~$T\mathcal S$ for any solution $(\ri,\lambda)$ of~$T^*\mathcal S$.
Plugging the ansatz for~$\mathcal N$ into the Noether condition,
we see that along the way there arise two types of
nonlocalities to deal with, $\mathrm D_x^{-1}(w^i_{\alpha l}\ri^l_x\lambda^\alpha)$ and
$\mathrm D_x^{-1}\mathrm D_t(w^i_{\alpha l}\ri^l_x\lambda^\alpha)$.
There are two ways to get rid of them.
A part of the nonlocalities are not genuinely nonlocal and, under certain conditions, can be expressed in terms of local variables.
For the other (genuine) nonlocalities, we should ensure that the collected coefficients of them vanish,
which results in the following constraints on the affinors~$w_\alpha$:
\[
w_\alpha=\mathrm e^{\ri^2-\ri^1}\mathrm{diag}\left(\Psi^\alpha_{\ri^1},-\Psi^\alpha_{\ri^2},\Phi^\alpha+\Psi^\alpha\right),
\]
where $\Phi^\alpha$ are arbitrary smooth functions of~$\ri^3$ and $\Psi^\alpha=\Psi^\alpha(\ri^1,\ri^2)$ are arbitrary solutions of the Klein--Gordon equation $\Psi^\alpha_{\ri^2}-\Psi^\alpha_{\ri^1}=2\Psi^\alpha_{\ri^1\ri^2}$.
The very same conditions ensure the non-genuine nonlocalities become local.
The form of the affinors~$w_\alpha$ is reminiscent of the hydrodynamic-type first-order generalized symmetries of the system~$\mathcal S$, see~\cite[Theorem~18]{OpanasenkoBihloPopovychSergyeyev2020a}.

The consideration until this point was valid for both $\mathfrak H^2$ and~$\mathfrak H^3$.
Now we prolong the operator~$\mathfrak H^2$, while the consideration for~$\mathfrak H^3$ is similar.
Solving the remaining determining equations,
we find that $\mathcal N$ is a Noether operator if and only if it is of the form
\begin{gather*}
\mathfrak N^2=\mathrm e^{\ri^2-\ri^1}\left(\mathrm{diag}\left(1,1,\mathrm e^{\ri^2-\ri^1}{\Theta}\right)\mathrm D_x+\frac12
\begin{pmatrix}
\ri^2_x-\ri^1_x &-\ri^1_x-\ri^2_x&-2\ri^3_x\\
\ri^1_x+\ri^2_x &\ri^2_x-\ri^1_x &2\ri^3_x\\
2\ri^3_x&-2\ri^3_x&f^{33}
\end{pmatrix}\right)+\mathfrak T_w^{ij},
\end{gather*}
where $f^{33}=2\mathrm e^{\ri^2-\ri^1}(\ri^2_x-\ri^1_x)\Theta+\bar\Theta$
for some smooth functions~$\Theta$ of~$\ri^3$ and~$\bar\Theta$ of~$(\ri^3,\mathrm e^{\ri^2-\ri^1}\ri^3_x)$.
In addition, the operator~$\mathfrak N^2$ is formally skew-adjoint
if and only if $\bar\Theta=\mathrm e^{\ri^2-\ri^1}\ri^3_x\Theta_{\ri^3}$,
and this is exactly the condition that makes~$\mathfrak N^2$ to be of hydrodynamic type.
Moreover, the Jacobi identity for the Poisson bracket associated with~$\mathfrak N^2$
is equivalent to the following collection of additional constraints on the parameter functions~$\Phi^\alpha$ and~$\Psi^\alpha$:
\begin{subequations}\label{IDFM:eq:NonlocalHamiltonianSystem}
\begin{gather}
\sum_{\alpha=1}^3\ \epsilon_\alpha (\Phi^\alpha+\Psi^\alpha)\Psi^\alpha_{\ri^1}=-\mathrm e^{\ri^1-\ri^2},\label{IDFM:eq:NonlocalHamiltonianSystemA}\\
\sum_{\alpha=1}^3\ \epsilon_\alpha(\Phi^\alpha+\Psi^\alpha)\Psi^\alpha_{\ri^2}=\mathrm e^{\ri^1-\ri^2}\label{IDFM:eq:NonlocalHamiltonianSystemB},\\
\sum_{\alpha=1}^3\ \epsilon_\alpha\Psi^\alpha_{\ri^1}\Psi^\alpha_{\ri^2}=0\label{IDFM:eq:NonlocalHamiltonianSystemC}.
\end{gather}
The system~\eqref{IDFM:eq:NonlocalHamiltonianSystem} follows from the condition
$R^{ij}_{\ \ kl}=\sum_\alpha\epsilon_\alpha(w^i_{\alpha l}w^j_{\alpha k}-w^i_{\alpha k}w^j_{\alpha l})$,
which is one of the conditions in Theorem~\ref{thm:NonlocalHamiltonianOpsWithAffinors}
that are required for the Jacobi identity.
The others of these conditions,
which are the commutativity of the affinors~$\omega_\alpha$,
the $j\leftrightarrow k$ symmetry of $\nabla_kw^i_{\alpha j}$
and the $i\leftrightarrow j$ symmetry of $g_{ik}w^k_{\alpha j}$, are then satisfied automatically.
\end{subequations}

Let us start with the prolongation of the operator~$\mathfrak H^2$.
In order to obtain its analogue~$\hat{\mathfrak H}^2_{\Theta,\alpha}$,
we have to construct all the solutions of the system on the parameter functions~$\Phi^\alpha$ and $\Psi^\alpha$
that consists of the equations~\eqref{IDFM:eq:NonlocalHamiltonianSystem},
three copies $\Psi^\alpha_{\ri^2}-\Psi^\alpha_{\ri^1}=2\Psi^\alpha_{\ri^1\ri^2}$ of the Klein--Gordon equation
and the equations $\Phi^\alpha_{\ri^1}=\Phi^\alpha_{\ri^2}=0$ and $\Psi^\alpha_{\ri^3}=0$
just meaning that $\Phi^\alpha=\Phi^\alpha(\ri^3)$ and $\Psi^\alpha=\Psi^\alpha(\ri^1,\ri^2)$.
Note that any permutation of the tuples $(\Phi^\alpha,\Psi^\alpha)$
and the gauge transformation $(\Phi^\alpha,\Psi^\alpha)\mapsto(\Phi^\alpha-c_\alpha,\Psi^\alpha+c_\alpha)$
with arbitrary constants~$c_\alpha$ are symmetry transformations of this system.

The equations~\eqref{IDFM:eq:NonlocalHamiltonianSystemA} and~\eqref{IDFM:eq:NonlocalHamiltonianSystemB} jointly integrate to
\begin{gather}\label{IDFM:eq:NonlocalHamiltonianSystemIntegr}
\sum_{\alpha=1}^3\epsilon_\alpha\left(\Phi^\alpha+\frac12\Psi^\alpha\right)\Psi^\alpha=\Omega(\ri^3)-\mathrm e^{\ri^1-\ri^2}
\end{gather}
for some function~$\Omega$ of~$\ri^3$,
while the equation~\eqref{IDFM:eq:NonlocalHamiltonianSystemC} is
a differential consequence of the system of~\eqref{IDFM:eq:NonlocalHamiltonianSystemIntegr}
and three copies $\Psi^\alpha_{\ri^2}-\Psi^\alpha_{\ri^1}=2\Psi^\alpha_{\ri^1\ri^2}$ of the Klein--Gordon equation.
Fixing values of~$(\ri^1,\ri^2)$, we derive that $\Omega=a_\alpha\Phi^\alpha+a_0$
with constants~$a_0$ and~$a_\alpha$.
Then the condition~\eqref{IDFM:eq:NonlocalHamiltonianSystemIntegr} takes the form
\begin{gather}\label{IDFM:eq:NonlocalHamiltonianequivCond}
\epsilon_\alpha\Phi^\alpha(\Psi^\alpha-a_\alpha)+\frac{\epsilon_\alpha}2\Psi^\alpha\Psi^\alpha+\mathrm e^{\ri^1-\ri^2}-a_0=0.
\end{gather}
Modulo the gauge transformation $(\Phi^\alpha,\Psi^\alpha)\mapsto(\Phi^\alpha-c_\alpha,\Psi^\alpha+c_\alpha)$,
we can set $a_\alpha$ to arbitrary required values. The equality~\eqref{IDFM:eq:NonlocalHamiltonianequivCond}
implies that all~$\Psi^\alpha$'s are linear in~$\mathrm{e}^{\ri^1-\ri^2}$. Indeed, let $d:=\dim\langle\Phi^1,\Phi^2,\Phi^3,1\rangle$,
where $1$ denotes the constant function of~$\ri^3$ with value~1.
It is obvious that the integer~$d$ is nonzero and not greater than four by its definition.
The equality $d=4$ contradicts~\eqref{IDFM:eq:NonlocalHamiltonianequivCond}.
Therefore, $d\in\{1,2,3\}$.
Consider each of these possible values of~$d$ separately.

\medskip\par\noindent
$\boldsymbol{d=3.}$
Up to permutations of the pairs $(\Phi^\alpha,\Psi^\alpha)$,
we have $\Phi^3=b_i\Phi^i+b_0$ for some constants~$b_0$ and~$b_i$.
Here and in what follows the index~$i$ runs from~1 to~2.
We substitute the above expression for~$\Phi^3$ into~\eqref{IDFM:eq:NonlocalHamiltonianequivCond}
and split the obtained equation with respect to~$\Phi^i$,
which gives
$\Psi^i=-\epsilon_3\epsilon_ib_i(\Psi^3-a_3)+a_i$
and, after setting $a_i+\epsilon_3\epsilon_ib_ia_3=0$ and denoting $\tilde a_0:=a_0+a_3\epsilon_3b_0$,
\[
\frac12(\epsilon_ib_ib_i+\epsilon_3)(\Psi^3)^2+\epsilon_3b_0\Psi^3+\mathrm e^{\ri^1-\ri^2}-\tilde a_0=0.
\]
Since the function~$\Psi^3$ satisfies the Klein--Gordon equation $2\Psi^3_{\ri^1\ri^2}=\Psi^3_{\ri^2}-\Psi^3_{\ri^1}$,
the last equation implies $(\epsilon_ib_ib_i+\epsilon_3)\Psi^3_{\ri^1}\Psi^3_{\ri^2}=0$,
and thus all $\Psi^\alpha$'s are linear in~$\mathrm{e}^{\ri^1-\ri^2}$.

\medskip\par\noindent
$\boldsymbol{d=2.}$
Up to permutations of $(\Phi^\alpha,\Psi^\alpha)$,
we can assume that $\dim\langle\Phi^3,1\rangle=2$.
Then $\Phi^i=b_{i1}\Phi^3+b_{i0}$ for some constants~$b_{i0}$ and~$b_{i1}$.
We substitute these expressions for~$\Phi^i$ into~\eqref{IDFM:eq:NonlocalHamiltonianequivCond}
and split the obtained equation with respect to~$\Phi^3$,
which gives
\begin{gather*}
\Psi^3=-\epsilon_3\epsilon_ib_{i1}(\Psi^i-a_i)+a_3,\\
\epsilon_ib_{i0}(\Psi^i-a_i)+\frac{\epsilon_i}2\Psi^i\Psi^i
+\frac{\epsilon_3}2(-\epsilon_3\epsilon_ib_{i1}(\Psi^i-a_i)+a_3)^2
+\mathrm e^{\ri^1-\ri^2}-a_0=0.
\end{gather*}
The second equation implies that the solutions~$\Psi^1$ and~$\Psi^2$ of the Klein--Gordon equation
are functionally dependent and thus this dependence can only be affine,
$\Psi^i=c_{i1}\Upsilon+c_{i0}$ for some constants~$c_{i0}$ and~$c_{i1}$ with $(c_{11},c_{21})\ne(0,0)$
and a nonconstant solution $\Upsilon=\Upsilon(\ri^1,\ri^2)$ of the same equation.
Under the derived representation for~$\Psi^i$,
the second equation reduces to an equation with respect to~$\Upsilon$,
\begin{gather*}
\epsilon_ib_{i0}(c_{i1}\Upsilon+c_{i0}-a_i)
+\frac{\epsilon_i}2(c_{i1}\Upsilon+c_{i0})^2
+\frac{\epsilon_3}2(-\epsilon_3\epsilon_ib_{i1}(c_{i1}\Upsilon+c_{i0}-a_i)+a_3)^2
+\mathrm e^{\ri^1-\ri^2}-a_0\\
=
\epsilon_ib_{i0}c_{i1}\Upsilon+\epsilon_ib_{i0}(c_{i0}-a_i)
+\frac{\epsilon_i}2c_{i1}^2\Upsilon^2+\epsilon_ic_{i1}c_{i0}\Upsilon+\frac{\epsilon_i}2c_{i0}^2
\\\qquad{}
+\frac{\epsilon_3}2(\epsilon_ib_{i1}c_{i1})^2\Upsilon^2
-\epsilon_ib_{i1}c_{i1}(a_3-\epsilon_3\epsilon_{i'}b_{i'1}(c_{i'0}-a_{i'}))\Upsilon
+\frac{\epsilon_3}2(a_3-\epsilon_3\epsilon_{i'}b_{i'1}(c_{i'0}-a_{i'}))^2
\\\qquad{}
+\mathrm e^{\ri^1-\ri^2}-a_0
=0.
\end{gather*}
Since $\Upsilon$ satisfies the Klein--Gordon equation,
this equation cannot be quadratic, which proves our point.

\medskip\par\noindent
$\boldsymbol{d=1.}$
Then all $\Phi^\alpha$ are constants and can be set to be equal to zero modulo gauge symmetries.
Thus, the parameter functions~$\Psi^\alpha$ satisfy the system of three copies of the Klein--Gordon equation, $2\Psi^\alpha_{\ri^1\ri^2}=\Psi^\alpha_{\ri^2}-\Psi^\alpha_{\ri^1}$,
with the constraint
\begin{gather}\label{IDFM:eq:psi_eq}
\sum_{\alpha=1}^3 \epsilon_\alpha(\Psi^\alpha)^2=C-2\mathrm e^{\ri^1-\ri^2}.
\end{gather}
Compatibility problem for the above system is computationally difficult, and
therefore, we use the fact that the Klein--Gordon equation has a family of solutions $u^k=\mathrm e^{k\ri^1/2+k\ri^2/(1-2k)}$,
to look for a solution of the equation~\eqref{IDFM:eq:psi_eq} as a formal series
\[
\Psi^\alpha=\sum\limits_{k\in\mathbb N_0} c_{\alpha k}u^k,
\]
which is an adaptation of the Taylor series of~$\Psi^\alpha$.
\noprint{
Indeed, let $g(\tau,\rho)$ be an analytic solution of the Klein--Gordon equation
$2\tau g_{\tau\rho}=g_{\rho}-\tau g_\tau$. Then the coefficients of its Taylor expansion at $\tau=\tau_0$,
\[
g=\sum\limits_{k\in\mathbb N_0} f_k(\rho)(\tau-\tau_0)^k=
\sum\limits_{k\in\mathbb N_0} f_k(\rho) \sum\limits_{i=0}^k{k\choose i} \tau^i\tau_0^{k-i}=\sum\limits_{k\in\mathbb N_0} h_k(\rho)\tau^k,
\]
must satisfy
\[
2kh_k'-h_k'+kh_k=0,
\]
that is, $h_k=c_k\mathrm e^{k\rho^2/(1-2k)}$.
Changing the coordinates $\tau=\mathrm e^{\ri^1}$, $\rho=\ri^2$, $g=\Psi^\alpha$,
we get the desired series.
}
The equation~\eqref{IDFM:eq:psi_eq} then takes the form
\begin{gather}\label{IDFM:eq:psi_eq_ser}
\sum_{\alpha=1}^3 \epsilon_\alpha\left(\sum\limits_{k\in\mathbb N_0} c_{\alpha k}u^k\right)^2=C-2\mathrm e^{\ri^1-\ri^2}.
\end{gather}
It is straightforward to show that $u^ku^l=u^{k'}u^{l'}$ if and only if $(k,l)=(k',l')$ or $(k,l)=(l',k')$,
and therefore the equation~\eqref{IDFM:eq:psi_eq_ser} can further be rewritten as
\begin{gather*}
\sum_{\alpha=1}^3 \epsilon_\alpha\left(\sum\limits_{k,l\in\mathbb N_0} c_{\alpha k}u^k c_{\alpha l}u^l\right)
=C-2\mathrm e^{\ri^1-\ri^2}.
\end{gather*}
The ``regular'' terms ($(k,l)\notin\{(0,0),(0,1),(1,0)\}$) of the above series give rise to equalities
\[
\sum_{\alpha=1}^3 \epsilon_\alpha c_{\alpha k}c_{\alpha l}=0,
\]
which imply that $c_{\alpha k}=0$ for $k\neq0$ if all~$\epsilon$'s are equal,
and then the (1,0)-term gives a contradiction.
Therefore, not all $\epsilon$'s are equal and without loss of generality we can assume that $\epsilon_1=\epsilon_2=-\epsilon_3=1$.
This leads to the equality
\begin{gather}\label{IDFM:eq:coeffs_kl}
c_{1k}c_{1l}+c_{2k}c_{2l}-c_{3k}c_{3l}=0
\end{gather}
for all regular tuples $(k,l)$, and in particular,
$
c_{3 k}^2=c_{1 k}^2+c_{2 k}^2
$
for all $k\neq0$, which allows to deduce from the equality~\eqref{IDFM:eq:coeffs_kl} that
\[c_{1k}c_{2l}-c_{2k}c_{1l}=0,\]
meaning all regular tuples $(c_{1k},c_{2k})$ are proportional to each other, in particular, $(c_{1k},c_{2k})=\kappa(c_{11},c_{21})$.
Plugging it back into~\eqref{IDFM:eq:coeffs_kl} we have
\begin{gather*}
\kappa c_{31}^2-c_{31}c_{3l}=0,
\end{gather*}
and thus $(c_{1k},c_{2k},c_{3k})=\kappa(c_{11},c_{21},c_{31})$ unless $c_{31}\neq0$.
With this in mind, we can write the $(k,0)$-term as
\begin{gather*}
\kappa(c_{11}c_{10}+c_{21}c_{20}-c_{31}c_{30})=0,
\end{gather*}
which is true regardless of the value of $c_{31}$.
It is in a clear contradiction with the (1,0)-term
\begin{gather*}
c_{11}c_{10}+c_{21}c_{20}-c_{31}c_{30}=-1,
\end{gather*}
unless $\kappa=0$.
Thus, the Klein--Gordon solutions are again of the form $\Psi^\alpha=c_{\alpha0}+c_{\alpha1}\mathrm e^{\ri^1-\ri^2}$.
Recall that we can set $c_{\alpha0}=0$ by the gauge transformation, and thus
the initial condition~\eqref{IDFM:eq:NonlocalHamiltonianSystemIntegr} is
\begin{gather*}
\sum_{\alpha=1}^3\epsilon_\alpha\left(\Phi^\alpha+\frac12c_{\alpha1}\mathrm e^{\ri^1-\ri^2}\right)c_{\alpha1}\mathrm e^{\ri^1-\ri^2}=
\Omega(\ri^3)-\mathrm e^{\ri^1-\ri^2}.
\end{gather*}
In particular, $\epsilon_1=\epsilon_2=-\epsilon_3=1$, $\Omega(\ri^3)=0$ and
\begin{gather*}
c_{11}^2+c_{21}^2=c_{31}^2,\\
c_{11}\Phi^1+c_{21}\Phi^2-c_{31}\Phi^3=-1.
\end{gather*}

For the operator~$\mathfrak H^3$ consideration is similar. Its analogue of the system~\eqref{IDFM:eq:NonlocalHamiltonianSystem} is
\begin{gather*}
\sum_{\alpha=1}^3\ \epsilon_\alpha (\Phi^\alpha+\Psi^\alpha)\Psi^\alpha_{\ri^1}=\frac12(\ri^1+\ri^2+1)\mathrm e^{\ri^1-\ri^2},\\
\sum_{\alpha=1}^3\ \epsilon_\alpha(\Phi^\alpha+\Psi^\alpha)\Psi^\alpha_{\ri^2}=-\frac12(\ri^1+\ri^2-1)\mathrm e^{\ri^1-\ri^2},
\end{gather*}
which integrates to
\begin{gather*}
\sum_{\alpha=1}^3\ \epsilon_\alpha (\Phi^\alpha+\frac12\Psi^\alpha)\Psi^\alpha=\Omega(\ri^3)-\frac12(\ri^1+\ri^2)\mathrm e^{\ri^1-\ri^2}.
\end{gather*}
The essential difference to the $\mathfrak H_2$-case lies in the series expansion in the section $d=1$.
The Klein--Gordon equation $2u_{xy}=u_y-u_x$ admits a Lie symmetry $2x\p_x-2y\p_y+(x+y)u\p_u$,
whose characteristic is $\mathrm J[u]:=u(x+y)-2xu_x+2yu_y$.
We thus represent a solution of the Klein--Gordon equation as the formal series
\[
\Psi^\alpha=\sum\limits_{k\in\mathbb N_0} c_{\alpha k}v^k,
\]
where $v_k:=((1-2k)^2\ri^1+\ri^2)\mathrm e^{k\ri^1+k\ri^2/(1-2k)}$ is the image of the solution~$u^k$
under the aforementioned Lie symmetry, $v^k=(1-2k)\mathrm J[\mathrm e^{k\ri^1+k\ri^2/(1-2k)}]$.
\end{proof}

\begin{remark}\label{rem:ReciprocalTrans}
As discussed in the introduction, Hamiltonian operators of Ferapontov type can be obtained from Dubrovin--Novikov operators by applying a reciprocal transformation.
Thus, a reciprocal transformation associated with the solution $\Psi(\ri^1,\ri^2)=\mathrm e^{\ri^1-\ri^2}$ of the Klein--Gordon equation,
\[
\mathrm d\tilde x=\mathrm e^{\ri^1-\ri^2}(\mathrm dx-(\ri^1+\ri^2)\mathrm dt),\quad \mathrm d\tilde t=\mathrm dt
\]
maps the system~$\mathcal S$ to the hydrodynamic-type system
\begin{gather*}
\ri^1_{\tilde t}=-\mathrm e^{\ri^1-\ri^2}\ri^1_{\tilde x},\quad \ri^2_{\tilde t}=\mathrm e^{\ri^1-\ri^2}\ri^2_{\tilde x},\quad \ri^3_{\tilde t}=0,
\end{gather*}
which possesses the three families of local Hamiltonian operators of the Dubrovin--Novikov type,
\begin{gather*}
\tilde{\mathfrak H}^1_{\tilde \Theta}=\mathrm e^{\ri^1-\ri^2}
\left(\mathrm{diag}(-1,1,\mathrm e^{\ri^2-\ri^1}\tilde\Theta)\mathrm D_{\tilde x}-\frac12
\begin{pmatrix}
\ri^2_{\tilde x}-\ri^1_{\tilde x} & \ri^1_{\tilde x}-\ri^2_{\tilde x} & 0\\
\ri^2_{\tilde x}-\ri^1_{\tilde x} & \ri^1_{\tilde x}-\ri^2_{\tilde x} & 0\\
0&0& -\mathrm e^{\ri^2-\ri^1}{\ri^3_{\tilde x}}\tilde\Theta_{\ri^3}
\end{pmatrix}\right),\\
\tilde{\mathfrak H}^2_{\tilde \Theta}=\mathrm e^{\ri^1-\ri^2}
\left(\mathrm{diag}(1,1,\mathrm e^{\ri^2-\ri^1}\tilde\Theta)\mathrm D_{\tilde x}+\frac12
\begin{pmatrix}
\ri^1_{\tilde x}-\ri^2_{\tilde x}  & \ri^1_{\tilde x}+\ri^2_{\tilde x} & 0 \\
-\ri^1_{\tilde x}-\ri^2_{\tilde x} & \ri^1_{\tilde x}-\ri^2_{\tilde x} & 0 \\
0&0& \mathrm e^{\ri^2-\ri^1}{\ri^3_{\tilde x}}\tilde\Theta_{\ri^3}
\end{pmatrix}\right),\\
\tilde{\mathfrak H}^3_{\tilde \Theta}=\mathrm e^{\ri^1{-}\ri^2}
\mathrm{diag}(\ri^1,\ri^2,\mathrm e^{\ri^2{-}\ri^1}\tilde\Theta)\mathrm D_{\tilde x}
{+}\frac{\mathrm e^{\ri^1{-}\ri^2}}2
\begin{pmatrix}
\ri^1_{\tilde x}{+}\ri^1\ri^1_{\tilde x}{-}\ri^1\ri^2_{\tilde x}  & \ri^2\ri^1_{\tilde x}{+}\ri^1\ri^2_{\tilde x} & 0\\
{-}\ri^2\ri^1_{\tilde x}{-}\ri^1\ri^2_{\tilde x}     & \ri^2\ri^1_{\tilde x}{+}\ri^2_{\tilde x}{-}\ri^2\ri^2_{\tilde x} & 0\\
0&0& \mathrm e^{\ri^2{-}\ri^1}{\ri^3_{\tilde x}}\tilde\Theta_{\ri^3}
\end{pmatrix}
\end{gather*}
parameterized by an arbitrary function~$\tilde\Theta$ of~$\ri^3$.
\end{remark}

A similar example of prolonging a local Hamiltonian operator to a nonlocal one arises for $\p_y$-reduction of the (1+2)-dimensional shallow water model,
see~\cite[Theorem~4.4, p.~130]{Opanasenko2021}.

\section*{Acknowledgments}
We thank E. Ferapontov for useful discussions.
SO acknowledges financial support from Istituto Nazionale di Fisica Nucleare (INFN) by IS-CSN4
\emph{Mathematical Methods of Nonlinear Physics} (MMNLP), from the GNFM of Istituto Nazionale di Alta Matematica ``F. Severi''
(INdAM) and from Dipartimento di Matematica e Fisica ``E. De Giorgi'' of the Universit\`a del Salento.
It was also supported in part by the Ministry of Education, Youth and Sports of the Czech Republic (M\v SMT \v CR)
under RVO funding for I\v C47813059.
ROP expresses his gratitude for the hospitality shown by the University of Vienna during his long staying at the university.
The authors express their deepest thanks to the Armed Forces of Ukraine and the civil Ukrainian people
for their bravery and courage in defense of peace and freedom in Europe and in the entire world from russism.

\end{document}